\documentclass[acmtog, balance = false]{acmart}
\acmSubmissionID{734}

\usepackage{booktabs} %

\acmJournal{TOG}

\copyrightyear{2024}
\acmYear{2024}
\setcopyright{acmlicensed}\acmConference[SA Conference Papers '24]{SIGGRAPH Asia 2024 Conference Papers}{December 3--6, 2024}{Tokyo, Japan}
\acmBooktitle{SIGGRAPH Asia 2024 Conference Papers (SA Conference Papers '24), December 3--6, 2024, Tokyo, Japan}
\acmDOI{10.1145/3680528.3687650}
\acmISBN{979-8-4007-1131-2/24/12}

\usepackage{hyperref}
\usepackage{url}
\usepackage{wrapfig}
\usepackage{xcolor}
\usepackage{colortbl}
\usepackage{multirow}

\usepackage{wrapfig}
\usepackage{nicefrac}
\usepackage{mathtools}
\usepackage{graphicx}
\usepackage{booktabs} 
\usepackage{combelow} %
\usepackage{tabularx} %
\usepackage{colortbl} %

\usepackage{manfnt} %

\usepackage{caption}
\usepackage{subcaption}

\usepackage{amsmath}
\usepackage{amsfonts}
\usepackage{amsbsy}

\usepackage{cleveref}

\usepackage{algpseudocode}

\usepackage[ruled]{algorithm2e} %

\SetAlFnt{\small}
\SetAlCapFnt{\small}
\SetAlCapNameFnt{\small}
\SetAlCapHSkip{0pt}

\usepackage[utf8x]{inputenc}

\usepackage{xr}

\usepackage{listings}

\definecolor{codegreen}{rgb}{0,0.6,0}
\definecolor{codegray}{rgb}{0.5,0.5,0.5}
\definecolor{codepurple}{rgb}{0.58,0,0.82}
\definecolor{backcolour}{rgb}{0.95,0.95,0.92}

\lstdefinestyle{my_code_style}{
  backgroundcolor=\color{backcolour}, 
  commentstyle=\color{codegreen},
  keywordstyle=\color{magenta},
  numberstyle=\tiny\color{codegray},
  stringstyle=\color{codepurple},
  basicstyle=\ttfamily\footnotesize,
  breakatwhitespace=false,         
  breaklines=true,                 
  captionpos=b,                    
  keepspaces=true,                 
  numbers=left,                    
  numbersep=5pt,                  
  showspaces=false,                
  showstringspaces=false,
  showtabs=false,                  
  tabsize=2
}
\usepackage{xfrac}

\usepackage{float}

\newcommand{\refequ}[1] {Eq.~\eqref{equ:#1}}

\newcommand{\reffig}[1] {Fig.~\ref{fig:#1}}

\newcommand{\reftbl}[1] {Table~\ref{tbl:#1}}
\newcommand{\refsec}[1] {Sec.~\ref{sec:#1}}
\newcommand{\refapp}[1] {App.~\ref{app:#1}}

\newcommand{\reflem}[1] {Lemma~\ref{lem:#1}}

\newcommand{\st}{\text{s.t.}}

\newcommand{\dx}{\vu}

\newcommand{\vecFont}[1]{\mathbf{#1}}

\def\vb{{\vecFont{b}}}

\def\ve{{\vecFont{e}}}

\def\vg{{\vecFont{g}}}

\def\vu{{\vecFont{u}}}

\def\vx{{\vecFont{x}}}

\newcommand{\matFont}[1]{\mathbf{#1}}
\def\mA{{\matFont{A}}}

\def\mH{{\matFont{H}}}

\def\mP{{\matFont{P}}}

\newtheorem{theorem}{Theorem}[section]

\newtheorem{lemma}[theorem]{Lemma}

\newcommand{\changed}[1]{{#1}}
\colorlet{RED}{red} %

\newlength\savedwidth
\newcommand\whline[1]{\noalign{\global\savedwidth\arrayrulewidth
                               \global\arrayrulewidth #1} %
                      \hline
                      \noalign{\global\arrayrulewidth\savedwidth}}

\lstdefinelanguage{Pseudocode}%
  {morekeywords={abstract,break,case,catch,const,continue,do,else,elseif,%
      end,export,false,for,function,immutable,import,importall,if,in,%
      macro,module,otherwise,quote,return,switch,true,try,type,typealias,%
      using,while,either,or,max,abs,foreach},%
   sensitive=true,%
   alsoother={$},%
   morecomment=[l]\#,%
   morecomment=[n]{\#=}{=\#},%
}[keywords,comments,strings]%

\definecolor{derekblue}{RGB}{144,187,195}
\definecolor{darkderekblue}{RGB}{86,130,140}
\definecolor{verylightgray}{RGB}{245,245,245}

\begin{document}

\title{\changed{Trust-Region} Eigenvalue Filtering for Projected Newton}

\author{Honglin Chen}
\affiliation{%
  \institution{Columbia University}
  \country{United States of America}
  }
\email{honglin.chen@columbia.edu}

\author{Hsueh-Ti Derek Liu}
  \affiliation{%
    \institution{Roblox \& University of British Columbia}
    \country{Canada}}
  \email{hsuehtil@gmail.com}

\author{Alec Jacobson}
\affiliation{%
  \institution{University of Toronto \& Adobe Research}
  \country{Canada}}
\email{jacobson@cs.toronto.edu}

\author{David I.W. Levin}
\affiliation{%
  \institution{University of Toronto \& NVIDIA}
  \country{Canada}}
\email{diwlevin@cs.toronto.edu}

\author{Changxi Zheng}
  \affiliation{%
    \institution{Columbia University}
    \country{United States of America}}  
  \email{cxz@cs.columbia.edu}

\begin{teaserfigure}
\centering
  \includegraphics[width=0.96\linewidth]{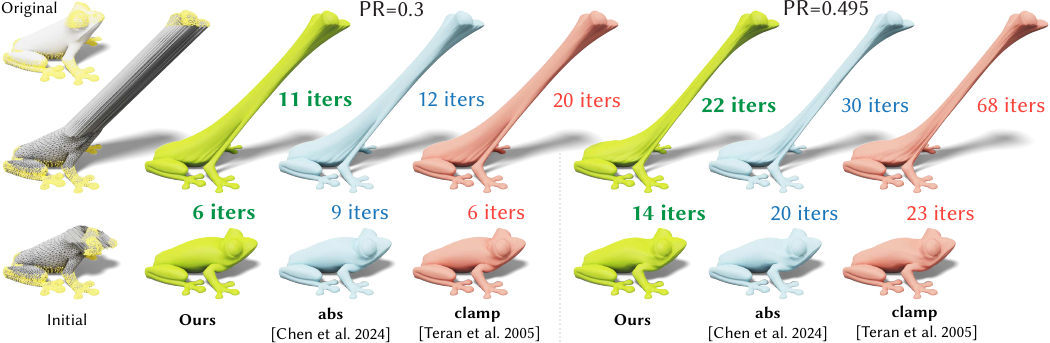}
  \caption{Our adaptive eigenvalue projection scheme stabilizes the projected Newton optimization of stable Neo-Hookean energy under high Poisson's ratios (PR) and large initial volume change while maintaining fast convergence in other cases.
  Here the fixed vertices are colored in yellow. Our method consistently outperforms the other approaches.
  }
  \label{fig:teaser}
\end{teaserfigure}

\begin{abstract}
We introduce a novel adaptive eigenvalue filtering strategy to stabilize and accelerate the optimization of Neo-Hookean energy and its variants under the Projected Newton framework.
For the first time, we show that Newton's method, Projected Newton with eigenvalue clamping and Projected Newton with absolute eigenvalue filtering can be unified \changed{using ideas from the generalized trust region method}. 
Based on the trust-region fit, our model adaptively chooses the correct eigenvalue filtering strategy to apply during the optimization.
Our method is simple but effective, requiring only two lines of code change in the existing Projected Newton framework.
We validate our model outperforms stand-alone variants across a number of experiments on \changed{quasistatic simulation of} deformable solids over a large dataset.
\end{abstract}

\maketitle

\vspace{-1mm}
\section{Introduction}
Hyperelastic simulations are indispensable in capturing the rich visual behaviors of deformable objects -- from highly compressible foam to nearly incompressible rubber and human tissues.
However, the numerical optimization of these energies can be notoriously challenging due to the non-convexity of the volume-preserving term.
This high non-convexity can lead to unstable and slow convergence in commonly applied Projected Newton optimizers, especially under high Poisson's ratios (near $0.5$) and large initial volume change.
This necessitates the use of eigenvalue filtering strategies to enforce the positive definiteness of the Hessian matrix and guarantee convergence of the optimization algorithm.

In this work we depart from the Projected Newton perspective and show that (somewhat surprisingly), we can unify existing eigenvalue clamping and filtering strategies under the \changed{\emph{generalized trust region}} formalism.
This motivates us to rethink the eigenvalue filtering step. Instead of \changed{\textit{a priori}} choosing eigenvalue filtering strategy ahead of time,  we show that a significant reduction in Newton iterations can be obtained by adaptively choosing the correct strategy based on how well a quadratic model fits the underlying nonlinear energy in a particular trust-region.
This leads to an extremely \emph{elegant} and \emph{effective} solution that requires only two lines of code change \changed{in} an existing Projected Newton solver to implement our more performant, trust-region-based approach (see Sec.\ref{lst:our_method}).

We demonstrate the effectiveness and efficiency of our method through a wide range of challenging examples, including different deformations, physical parameters, geometries, mesh resolutions, and elastic energies.
Through extensive experiments, we show that \changed{on average,} our adaptive method outperforms all other stand-alone variants in terms of stability and convergence speed \changed{in quasistatic simulations}.
Our method is robust across different mesh resolutions and Poisson's ratios and highly distorted and inverted elements, while maintaining the same per-iteration computational cost and simplicity of implementation.

To summarize, our contributions are as follows:

\begin{enumerate}
  \item We offer a theoretical analysis of nonconvex optimization of Neo-Hookean energy based on the generalized trust region framework. 
  For the first time, we show that Newton's method, Projected Newton with eigenvalue clamping and Projected Newton with absolute eigenvalue filtering can be unified \changed{based on insights from the generalized trust region method}.

  \item Based on our analysis, we propose a novel adaptive eigenvalue filtering strategy which correctly chooses the strategy to apply based on the trust-region fit.
  \changed{Our method is still formally a (regularized) Projected Newton, where the regularization parameter mimics the Lagrange multiplier arising in trust-region subproblems. }
  Our approach ensures stability under high Poisson's ratios and large initial volume change, while still maintaining fast convergence in all other regimes.
\end{enumerate}

\section{Related Work}\label{sec:related}

\paragraph{Projected Newton's method.}

Projected Newton's methods rely on projection strategies to restore the positive definiteness of the Hessian matrix (to guarantee convergence \changed{to a local minimum}).
As an example, \citet{Paternain2019nonconvex_newton} proposed projecting eigenvalues of the \emph{global} Hessian to a small positive value when its magnitude is small, and to its absolute value otherwise.
Due to the computational cost of a global eigendecomposition, in a finite element setting, these projection often operate on the element-wise Hessian submatrices. 
Choices of projection range from clamping negative eigenvalues to a small positive number (or zero)~\cite{irving05},  adding a diagonal~\cite{FuL16} or setting negative eigenvalues to their absolute values\changed{~\cite{chen2024abs_eig,num_opt2006,gill1981practical}}.
Some strategies eschewed eigenvalue filtering which allows for modifications to the global hessian.
For instance, \citet{longva2023pitfalls} either suggested using the original
Hessian matrix whenever it is positive definite or adding multiples of the mass matrix until Cholesky Factorization succeeds.

\paragraph{Trust-region method.}

Trust-region methods have been extensively studied in optimization literature \cite{Conn2000TrustRegion, num_opt2006} for its stability and robustness.
In some sense it can be viewed as the dual of line-search method: it first chooses a step size and then picks a step direction within it.
\citet{Sorensen1982newton} first proposed the trust region method by making a trust region modification to Newton's method. 
\citet{more1983trust_region_step} studied the strategy to pick the trust region step.
\citet{more1993generalizations} extended it to the generalized trust region method, which has been actively studied until today \cite{pong2014upper_lower_bound, wang2020generalized}.
\changed{Several works combined the idea of trust-region and line search method and propose the regularized Newton's method \cite{regularized_newton_2014,regularized_newton_2015,regularized_newton_2023}.}
In machine learning, \citet{dauphin2014saddle_free} introduced the saddle-free Newton method, which is a generalized trust region method with a first-order model and a quadratic constraint.

\paragraph{Neo-Hookean Energy}
Several recent works have focused on improving the stability of Neo-Hookean energy \cite{ogden1997non}.
To improve robustness to mesh inversion and rest-state stability, \citet{smith2018snh}
proposed the stable Neo-Hookean energy and applied their analysis to stabilize other hyperelastic variants \cite{kim22deformables,lin2022isoARAP}.
\citet{chen2024abs_eig} further showed the high non-convexity of the stable Neo-Hookean energy stems from high Poisson's ratio and large volume change, and proposed an eigenvalue filtering scheme to stabilize the optimization. 
\section{Background}

\begin{figure}
  \centering
  \includegraphics[width=0.96\linewidth]{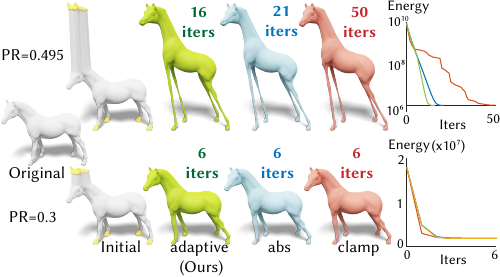}
  \includegraphics[width=0.96\linewidth]{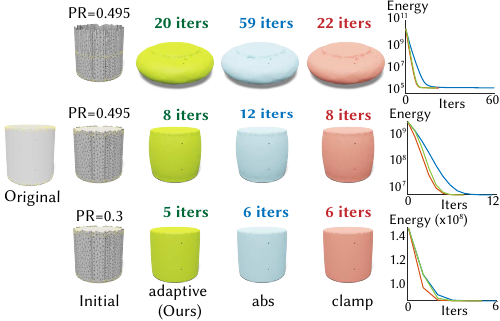}
  \caption{Our adaptive strategy stabilizes the optimization under large volume
  changes and high Poisson's ratios (PR), while maintaining fast convergence in other
  cases. The three columns in the middle 
  visualize the deformed meshes after the energy minimization converges in our method (green), 
  absolute eigenvalue projection (blue), and eigenvalue clamping (red).
  The number of iterations for each method until convergence is labeled on the top.
  The last column shows the energy convergence curves
  with respect to the number of iterations.} 
  \label{fig:stretch_compress}
  \vspace{-1mm}
\end{figure}
In many computer graphics applications, such as quasistatic hyperelastic simulation, an optimization problem arises:
\begin{align}\label{equ:min}
  \min_{\vx}\ f(\vx)\;,
\end{align}
where $f$ denotes a (twice differentiable) energy with parameters $\vx$. 
We start by briefly reviewing two seemingly separate methods for solving this optimization problem\textemdash namely,
the \emph{projected Newton's method} and \emph{trust-region method}\textemdash to lay out the foundation for our contributions,
which revolve around a formal connection between these two methods (which we will develop in \refsec{method}).
Readers familiar with these methods may directly jump to \refsec{method} instead.

\paragraph{Newton's method.}
In many cases, one can solve \refequ{min} with Newton's method by taking the second-order Taylor expansion of $f(\vx)$:
\begin{align}\label{equ:taylor_approximation}
  \tilde{f}(\vx + \dx) \approx f(\vx) + \vg^\top \dx + \frac{1}{2} \dx^\top \mH \dx,
\end{align}
where $\vg = \nabla f(\vx)$ and $\mH = \nabla^2 f(\vx)$ are the gradient and
the Hessian of $f$, respectively. Then, the optimal update $\dx$ based on the quadratic approximation $\tilde{f}$ can be obtained by solving a linear system
\begin{align}\label{equ:step}
  \dx = - \mH^{-1} \vg.
\end{align}
It is assumed that $\dx$ is an update direction that leads to a decrease of the energy $f$, and thus iterations between the above two steps will converge \changed{to a local minimum}.
More precisely, this assumption requires that 
\textbf{1)} the Hessian $\mH$ to be \emph{positive definite} and \textbf{2)} 
the second-order expansion $\tilde{f}$ being a \emph{close} quadratic approximation of the energy function $f$.
Unfortunately, these requirements do not always hold in practice.

\subsection{Projected Newton's Method}\label{sec:projected_newton}

Regarding the requirement 1), if $\mH$ is indefinite or negative definite, even if $\tilde{f}$ is a close approximation of $f$, the update $\dx$ may still be an \emph{ascent} direction,
because their \emph{critical points} $\nicefrac{\partial \tilde{f}}{\partial \dx} = 0$ are not at the minimum of $\tilde{f}$.

As a response to this situation,
the \emph{Projected Newton's method} approximates $\mH$ in \refequ{taylor_approximation} with a positive definite matrix. 
A straightforward approach performs eigenvalue analysis on
$\mH$, \changed{modifies} the negative eigenvalues \changed{to positive ones}, and then reconstructs a \changed{positive definite} matrix using the updated eigenvalues. Yet, 
for large matrices, this process is computationally intractable.

\begin{figure}[t]
  \centering
  \includegraphics[width=0.96\linewidth]{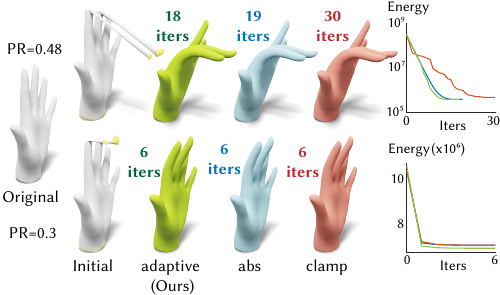}
  \includegraphics[width=0.96\linewidth]{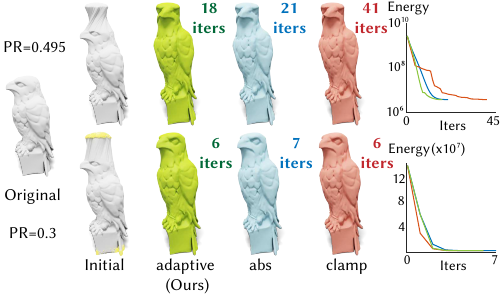}
  \caption{Two other examples in experiments similar to \reffig{stretch_compress}.
  Here we compare our method with others under bending (top) and twisting (bottom) deformations.  
  In all cases, our method requires the least number of iterations to converge.} 
  \label{fig:bend_twist}
  \vspace{-2mm}
\end{figure}

\paragraph{Eigenvalue clamping.}
In finite element simulations, a more efficient approach exists:
Notice that the (global) Hessian $\mH$ is assembled
by summing up the Hessian matrix $\mH_i$ from each finite element $i$, that is,
\begin{align}
  \mH = \sum_i \mP_i^\top \mH_i \mP_i,
\end{align}
where $\mP_i$ denotes the selection matrix that maps the per-element indices to the global indices.
Leveraging this property, \citet{irving05} proposed to project each
per-element Hessian $\mH_i$ to a positive definite one $\mH_i^+$. 
Their approach clamps the eigenvalues $\lambda_k$ of each $\mH_i$ to \changed{zero (or a small
positive value $\epsilon$)} using
\begin{align}\label{equ:clamping}
  \lambda_k^+ & = 
  \begin{cases}
    \epsilon & \text { if } \lambda_k \leq \epsilon, \\
    \lambda_k & \text { otherwise,}
    \end{cases}
\end{align}
and computes $\mH_i^+$ with the clamped eigenvalues.
This approach requires only a tiny eigenvalue analysis for each finite element, all of
which can be computed in parallel, and it guarantees the reconstructed global
Hessian $\mH^+$ from 
\begin{align}
  \mH^+ = \sum_i \mP_i^\top \mH_i^+ \mP_i,
\end{align}
to be positive definite \cite{Rockafellar1970convex}, leading to a more stable Newton's solver for minimizing $f$.
\changed{Even though the negative eigenvalues of per-element Hessian are clamped to zero, the additive contributions of neighboring elements and boundary conditions always lead to a positive-definite global Hessian~\cite{irving05}. }

\paragraph{Absolute eigenvalue projection.}
For a highly nonconvex energy landscape\textemdash such as one that arises in simulation of nearly incompressible materials, 
\changed{eigenvalue clamping (\refequ{clamping}) could lead to 
unstable search directions and an excessively large number of iteration steps \cite{chen2024abs_eig}. }

This issue stems from the violation of the requirement 2) stated above.
For a highly nonconvex energy function $f$,
the quadratic approximation $\tilde{f}$ tends to be a poor approximation of $f$.
\citet{chen2024abs_eig} further proposed a remedy for this issue by projecting the negative eigenvalues 
of each $\mH_i$ to their absolute values:
\begin{align}\label{equ:abs}
  \lambda_k^+ = \left| \lambda_k \right|.
\end{align}
They demonstrated that this absolute-value projection strategy leads to orders of magnitude speed-up in the best case. 
But in other nearly convex cases (e.g., materials involving minor volume changes), this strategy can slightly damp the convergence along the negative curvature directions.

In short, how to choose the projection strategy for any given simulation
scenario remains to be an open question. 

\subsection{(Generalized) Trust-Region Method} \label{sec:trust_region_background}

Separate from the projected Newton's method, \emph{Trust-region methods}~\cite{Sorensen1982newton,Conn2000TrustRegion} were proposed in response to the violation 
of the requirement 2) stated above.
The key idea here is to identify a \emph{radius} around the current solution $\vx$ where the quadratic function $\tilde{f}$ is a sufficiently 
tight approximation of $f$. 
Then, the update $\dx$ (including both the update direction and the step size) is
computed to be the optimal solution within this ``trust-worthy''
region~\cite{more1983trust_region_step}.  Formally, the trust region
method can be written as a constrained optimization problem:
\begin{equation}  
  \begin{aligned} 
  \label{equ:trust_region_method}
  \min_\dx\quad &\tilde{f}(\vx + \dx) = f(\vx) + \vg^\top \dx + \frac{1}{2} \dx^\top \mH \dx\\
  \st\quad &\| \dx \|^2 = \dx^\top \dx \leq \Delta
  \end{aligned}
\end{equation}
where $\Delta \in \mathbb{R}$ denotes the squared ball radius of the trust
region. A key component in trust region methods lies in dynamic adjustment of the
trust region size $\Delta$ to ensure sufficiently small error coming from the
quadratic approximation.

Vanilla trust region methods measure the radius based on the \emph{Euclidean
ball} distance $\| \dx \|^2 = \dx^\top \dx$. But in practice, many problems
favor measuring distance by an underlying (Riemannian or non-Riemannian) metric.
In other words, the bound should adapt to the underlying energy landscape.
This observation motivates the developement of the \emph{generalized trust region method}
\cite{more1993generalizations, pong2014upper_lower_bound} which replace the
Euclidean distance $\| \dx \|$ constraint by a quadratic measure: 
\begin{equation}  
  \begin{aligned} 
    \label{equ:generalized_trust_region_method}
  \min_\dx\quad &\tilde{f}(\vx + \dx) = f(\vx) + \vg^\top \dx + \frac{1}{2} \dx^\top \mH \dx\\
  \st\quad &\Delta_\text{lower} \leq \dx^\top \mA \dx + \dx^\top \vb \leq \Delta_\text{upper}.
  \end{aligned}
\end{equation}
where $\mA$ is a symmetric matrix that could be non-positive definite and $\Delta_\text{lower}/\Delta_\text{upper}$ denote the lower/upper bound of the trust region.

\section{Our Method}\label{sec:method}

The efficacy of eigenvalue filtering strategies discussed in \refsec{projected_newton} depends on specific simulation scenarios~\cite{chen2024abs_eig}. 
Interestingly, we discover that these strategies are fundamentally connected, and they can be viewed in an unified framework,
which in turn enables automatic adaptation of eigenvalue filtering strategy during the Newton's iterations.
We now describe this unified framework.

\subsection{Per-element Generalized Trust-region \changed{Newton} Method}\label{sec:gen_trm}
The starting point of our development is the generalized trust-region method (\refsec{trust_region_background}).
Consider a specific instance of the generalized trust-region problem~\eqref{equ:generalized_trust_region_method}, wherein
the lower and upper bounds have the same value but with opposite signs (i.e., $\Delta_\text{upper} = -\Delta_\text{lower} = \Delta$).
This trust-region problem has the following form:
\begin{equation}  
  \begin{aligned} 
\underset{\dx}{\min } \quad & f(\vx)+\vg^\top \dx + \dx^{\top} \mH \dx\\
\text {s.t.}  \quad & -\Delta \leq  \dx^\top \mA \dx + \dx^\top \vb  \leq \Delta.
\end{aligned} 
\end{equation}
where the energy is the quadratic approximation of a simulation energy $f$ (see \refequ{taylor_approximation}). We use $\vg$ to denote the gradient and $\mH$ to denote the Hessian. 
If we choose $\vb = 0, \mA = \mH$,  the two inequalities (lower and upper bounds) can be re-written as a single inequality with the absolute value. 
These specifications lead to the following instance of generalized trust-region problem:
\begin{equation}  
  \begin{aligned} 
\underset{\dx}{\min } \quad & f(\vx)+\vg^\top \dx + \dx^{\top} \mH \dx\\
\text {s.t.}  \quad & | \dx^{\top} \mH \dx | \leq \Delta .
\end{aligned} 
\label{equ:trust_region_newton_global}
\end{equation}

\begin{figure}[t]
  \centering
  \includegraphics[width=\linewidth]{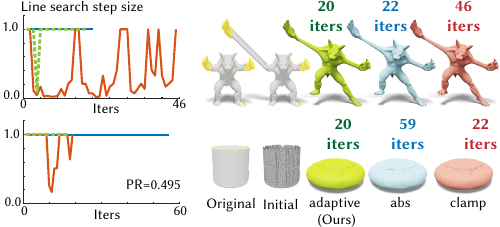}
  \caption{We demonstrate the effectiveness of our descent direction through the line search step size during the optimization.
  Our trust-region based adaptive scheme often requires only few iterations of line search to find a suitable step size. } 
  \label{fig:line_search_iter}
  \vspace{-2mm}
\end{figure}

In finite element simulation, the Hessian $\mH$ can often be constructed by
summing up per-element Hessians $\mH_i$ from each tetrahedral element $i$. 
Leveraging this property, we can re-write the constraint in
\eqref{equ:trust_region_newton_global} using the following lemma:

\begin{lemma}
Let $|\mA_e|$ be the matrix obtained by performing per-element absolute eigenvalue projection of $\mA$, i.e., $ |\mA_e| = \sum_i \mP_i^{\top} |\mA_i| \mP_i $, where $|\mA_i|$ are the matrix obtained by taking the absolute value of each of the eigenvalues of $\mA_i$.
Then it holds that $\left|\vx^{\top} \mA \vx\right|  \leq \vx^{\top} |\mA_e| \vx $.
\label{lem:abs_per_element_constraint}
\end{lemma}

\begin{proof} 
See \refapp{proof_lemma_1}.
\end{proof}

By applying \reflem{abs_per_element_constraint}, we have the following relationship: 
\begin{align}
   \dx^{\top}  |\mH_e| \dx \leq \Delta \implies |\dx^{\top}  \mH \dx| \leq \Delta,
\end{align}
from which \refequ{trust_region_newton_global} can be re-written as
\begin{equation}  
  \begin{aligned} 
\underset{\dx}{\min } \quad & f(\vx)+\vg^\top \dx + \dx^{\top} \mH \dx\\
\text {s.t.}  \quad & \dx^{\top}  |\mH_e| \dx \leq \Delta .
\end{aligned} 
\label{equ:trust_region_newton_hybrid_per_element}
\end{equation}  
This reformulation is a more conservative version of \refequ{trust_region_newton_global}; it
``tightens'' the size of trust region, but still ensures that
the search radius imposed by the constraint satisfies the original problem \refequ{trust_region_newton_global}.

During Newton's iterations, the trust-region size $\Delta$ may change.
When $\Delta$ is sufficiently large, the minimum point of the unconstrained quadratic approximation is within the trust region,
and this step is the same as the original Newton step.
But more often, the unconstrained quadratic minimizer lies outside the trust region. Then, the constrained minimizer will occur on the 
trust-region boundary.

Suppose $\Delta$ is sufficiently small so that the constrained minimizer occurs on the trust-region boundary.
In this situation,
we can replace the inequality constraint in \refequ{trust_region_newton_hybrid_per_element} with an equality constraint $\dx^{\top}  |\mH_e| \dx = \Delta$, 
and use Lagrangian multipliers to obtain its solution:
\begin{equation}  
  \begin{aligned} 
\vg + 2 \mH \dx + 2 \lambda |\mH_e| \dx = 0. 
\end{aligned} 
\label{equ:lagrangian_trust_region_newton_per_element}
\end{equation}  
Here, the Lagrangian multiplier $\lambda$ must satisfy $\lambda \ge 0$. 
It can be shown that $\lambda = 0$ corresponds to the original Newton
step \refequ{step} after a line search (see more details in \refsec{unprojected_newton_step}).  
In other words, although we arrive \refequ{lagrangian_trust_region_newton_per_element} by assuming 
the constrained minimizer occurs on the trust-region boundary, 
\refequ{lagrangian_trust_region_newton_per_element} itself also covers the case where
the unconstrained minimizer is within the trust region.

\changed{Solving for \refequ{lagrangian_trust_region_newton_per_element}} gives a Newton step of the following form (up to a scalar folded into line search):
\begin{equation}  
  \begin{aligned} 
 \dx & = - (\mH + \lambda |\mH_e| )^{-1} \vg  \\
  & = - \left((1+\lambda) \left(\frac{1}{1+\lambda}\mH + \frac{\lambda}{1+\lambda} |\mH_e| \right) \right)^{-1} \vg
\end{aligned} 
\label{equ:descent_direction_hybrid_per_element}
\end{equation}  

With a change of variable $w = \frac{\lambda}{1+\lambda}$ (and folding $(1+\lambda)$ into the line search), we can see that the step is a linear combination of the unprojected Newton step and the per-element absolute eigenvalue projected Newton step:
\begin{equation}  
  \begin{aligned} 
    \dx = - ((1-w) \mH + w |\mH_e| )^{-1} \vg
  \end{aligned} 
\label{equ:descent_direction_hybrid_per_element_weighted}
\end{equation}  

As the Lagrangian multiplier $\lambda$ must satisfy $\lambda \ge 0$, the weight $w$ must satisfy $0 \le w \le 1$.
To ensure the resulting Newton step is a descent direction, we also require $(1-w) \mH + w |\mH_e|$ to be positive definite.

Our derivation above can be \changed{interpreted} as a generalized trust-region
Newton method~\cite{num_opt2006, Boyd_Vandenberghe_2004} with the per-element
extension, equality constraint enforcement and an additional line search, 
\changed{or a regularized Newton's method~\cite{regularized_newton_2014,regularized_newton_2015,regularized_newton_2023} regularized by a generalized trust region.} 
These extensions retain the same computational cost as the original
Newton iteration while benefiting from the stability of the
trust-region based method.

\begin{figure}[t]
  \centering
  \includegraphics[width=\linewidth]{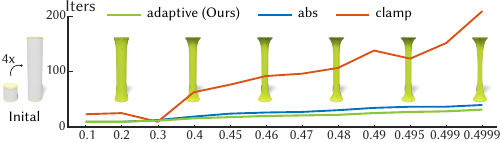}
  \caption{Here, the cylinder is stretched to 4$\times$ of its rest length.
  We increases the cylinder's Poisson's ratio (x-axis) and test the number of iterations
  (y-axis) needed for convergence in our method and the comparing methods.
  Across the entire range of Poisson's ratios, out method consistently outperforms
  others.
  } 
  \label{fig:different_pr}
  \vspace{-2mm}
\end{figure}
\subsection{Generalized Trust-region Newton Step}

Perhaps surprisingly, different choices of $w$ in
\refequ{descent_direction_hybrid_per_element_weighted} lead to different
pre-existing eigenvalue filtering strategies.
Here, we enumerate three typical choices of $w$ and their corresponding eigenvalue filtering strategies:

\subsubsection{Unprojected Newton step}\label{sec:unprojected_newton_step}

When $w = 0$, i.e., $\lambda = 0$, and the positive-definiteness constraint is satisfied, the step reduces to the traditional (unprojected) Newton step:
\begin{equation}  
  \begin{aligned} 
 \dx = - \mH^{-1} \vg
\end{aligned} 
\label{equ:descent_direction_original}
\end{equation}

\subsubsection{Projected Newton with eigenvalue clamping \cite{irving05}}
\label{sec:clamped_newton_step}

\begin{lemma}
  \changed{Let $\mA_e^+$ be the matrix obtained by performing eigenvalue clamping (with threshold 0) of each $\mA_i$, i.e., $ \mA_e^+ = \sum_i \mP_i^{\top} \mA_i^+ \mP_i $, where $\mA_i^+$ are the matrix obtained by clamping each eigenvalue of $\mA_i$ to $0$.
  Then it holds that $ \dx^{\top} \mA \dx + \dx^{\top} |\mA_e| \dx = 2 \dx^{\top} \mA_e^+ \dx $. }
  \label{lem:hybrid_sum}
\end{lemma}
  
\begin{proof} 
See \refapp{proof_lemma_3}.
\end{proof}
When $w = 0.5$, i.e., $\lambda = 1$, using \reflem{hybrid_sum}, the step reduces to the projected Newton step with eigenvalue clamping:
\begin{equation}  
  \begin{aligned} 
    \dx & = - (\frac{1}{2}\mH + \frac{1}{2}|\mH_e| )^{-1} \vg \\
    & = - (\mH_e^+)^{-1} \vg
\end{aligned} 
\label{equ:descent_direction_clamp_per_element}
\end{equation}

\subsubsection{Projected Newton with absolute eigenvalue projection \cite{chen2024abs_eig}}
\label{sec:abs_newton_step}

When $w = 1$, i.e., $\lambda \rightarrow \infty $, the step reduces to the projected Newton step with absolute eigenvalue projection:
\begin{equation}  
  \begin{aligned} 
 \dx = - | \mH_e |^{-1} \vg
\end{aligned} 
\label{equ:descent_direction_abs_per_element}
\end{equation}

Note that is also equivalent to using the first-order model and quadratic constraint with Hessian distance measure:
\begin{equation}  
  \begin{aligned} 
\underset{\dx}{\min } \quad & f(\vx)+\vg^\top \dx \\
\text {s. t.} \quad & \dx^{\top}  |\mH_e| \dx \leq \Delta .
\end{aligned} 
\label{equ:trust_region_newton_abs_per_element}
\end{equation}  
which, following a derivation similar to \cite{dauphin2014saddle_free}, leads to a step in the form \refequ{descent_direction_abs_per_element}.

\begin{figure}[t]
  \centering
  \includegraphics[width=\linewidth]{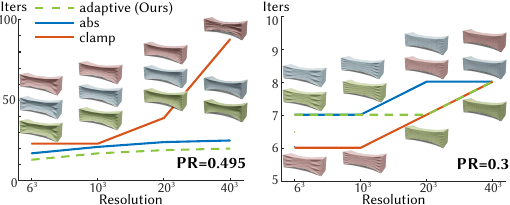}
  \caption{We plot the number of iterations needed for convergences with respect to 
  different mesh resolutions.
  Note that in the simple case (bottom plot), our method may takes one more Newton iteration than eigenvalue clamping, 
  simply because our method always starts from the absolute eigenvalue
  projection in its first Newton iteration.} 
  \label{fig:different_resolution}
  \vspace{-1mm}
\end{figure}

\subsection{Adaptive Per-element Projection}
\label{sec:adaptive_hybrid_projection}

Now, how should we choose the weight $w$ in \refequ{descent_direction_hybrid_per_element_weighted} to achieve stable energy descent?
To address this question, 
we seek inspiration from how the trust-region size is adjusted in trust-region method literature~\cite{num_opt2006}.
A typical strategy is to increase the trust region size when original function is well-approximated by the quadratic form, and decrease the trust region size otherwise.

In light of this, we choose $w$ based on how well the quadratic form (in \refequ{trust_region_newton_hybrid_per_element}) approximates the original energy function $f$:
\begin{enumerate}
\item If the quadratic form approximates the original energy well, then it is safe to allow a large trust-region size.
In the most extreme case where an infinite trust-region size is used, we can just remove the trust-region constraint in \refequ{trust_region_newton_hybrid_per_element}, which corresponds to
the unprojected Netwon step in \refsec{unprojected_newton_step} (i.e., $w=0$). However, to apply $w=0$, we need to ensure the global Hessian to be PSD. Therefore, we slightly 
relax it and apply $w=0.5$, which corresponds to eigenvalue \changed{clamping}\textemdash the strategy that ensures a PSD Hessian.

\item If the quadratic fitting quality is inadequate, then a small trust region size should be used.
In the most extreme case where a near-zero trust region size is used, we will have $\lambda \rightarrow \infty$ and thus the approximation model falls back to the first-order model, which corresponds to the absolute eigenvalue projection strategy.
This ensures sufficient regularization along the negative eigenvalue directions when the quadratic approximation is far from the original function.
\end{enumerate}

Therefore, we look for a ratio to quantify the agreement between the quadratic model and the original function, and adaptively adjust the weight $w$ based on this ratio.
In the trust region literature \cite{Sorensen1982newton,Conn2000TrustRegion}, a classic choice to measure this agreement is to use the ratio 
\begin{equation}  
  \begin{aligned}
    \rho = \frac{f(\vx) - f(\vx + \dx)}{\tilde{f}(\vx) - \tilde{f}(\vx + \dx)}
  \end{aligned}
  \label{equ:trust_region_newton_ratio}
\end{equation}
where $\tilde{f}(\vx) = f(\vx) + \vg^\top \vx + \frac{1}{2} \vx^\top \mH \vx$ is the quadratic approximation of $f(\vx)$.
When $\rho$ is close to $1$, the quadratic model is a good approximation of the original function, and a large trust region size can be allowed.
When $\rho$ is far from $1$ (e.g., close to zero or even negative), the quadratic model cannot accurately estimate the local landscape of the original function, and a small trust region size should be used.

Inspired by this, we propose to use the ratio $\rho$ to adaptively adjust the weight $w$ in \refequ{descent_direction_hybrid_per_element_weighted} based on the quality of the quadratic approximation:
\begin{align} 
  w
  =
  \left\{\begin{matrix} 
    0.5, & |\rho - 1| \le \epsilon \\
    1, & |\rho - 1| > \epsilon
  \end{matrix}\right. ,
  \label{equ:adaptive_weight}
\end{align}
where $\epsilon$ is a small positive constant between $0.01$ and $0.1$. 

\changed{
Our discrete $w$ is a practical design choice for computational efficiency.
Computing the optimal $w$ from \refequ{trust_region_newton_hybrid_per_element} would require solving a trust-region subproblem per Newton iteration, which involves an expensive iterative solve (see Sec 4.3 of \cite{num_opt2006}). 
Therefore, instead of solving the exact trust-region subproblems, we opt to focus on a discrete set of classical $w$ and pick one based on the trust region ratio.
}

In other words, let $\Lambda$ be the eigenvalues of the per-element Hessian, we \changed{perform} the following element-wise projection: 
\begin{align}
  \Lambda^+
  =
  \left\{\begin{matrix} 
    \max(\Lambda, 0), & |\rho - 1| \le \epsilon \\
    |\Lambda|, & |\rho - 1| > \epsilon
  \end{matrix}\right. ,
  \label{equ:adaptive_eigenvalue_projection}
\end{align}
Our method adaptively chooses between the absolute eigenvalue projection and eigenvalue clamping strategy based on the current energy landscape, ensuring sufficient regularization in the highly nonconvex regime while maintaining fast convergence in the (nearly) convex scenario.

For simplicity, we always start from the clamping strategy $|\Lambda|$ for the first Newton iteration, and then use the trust region ratio $\rho_k$ computed using $\vx_k$ and $\vx_{k-1}$ to pick the projection strategy for the $k$-th Newton iteration.

\begin{figure}[t]
  \centering
  \includegraphics[width=\linewidth]{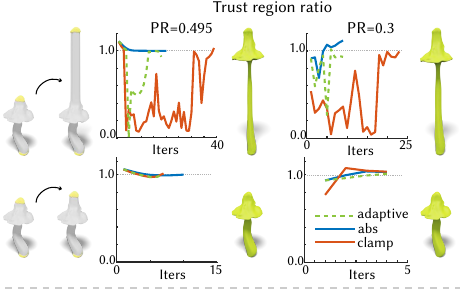}
  \includegraphics[width=\linewidth]{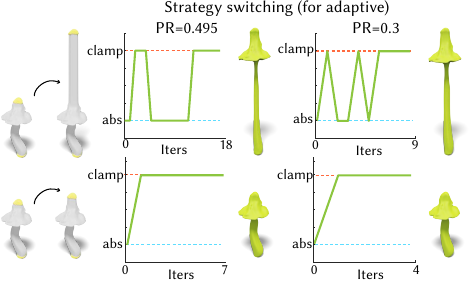}
  \vspace{-1mm}
  \caption{We visualize the trust-region ratio $\rho$ \changed{(for all strategies) and strategy switching trend (for adaptive)} during the optimization
  under different deformations and Poisson's ratios.  While the trust region
  ratio can change drastically in the early stage of optimization for the highly
  nonconvex energy landscape\textemdash e.g., \changed{large deformation and high Poisson's ratio}, it 
  usually approaches $1$ near convergence. 
  \changed{Thus the adaptive strategy often starts from abs strategy for the first few iterations, alternates between abs and clamp during the optimization, and falls back to clamp near convergence. 
  On the contrary, solely using the abs strategy may excessively damp the descent direction when the quadratic approximation is still acceptable (e.g., small PR or volume change) or later becomes good (e.g., near convergence). }
  } 
  \label{fig:trust_region_ratio_mushroom_stretch}
\end{figure}

\emph{Implementation.}
Our method requires only two lines of code change in the existing Projected Newton framework with line search.
For completeness, we provide the full algorithm of our adaptive eigenvalue projection as below:

\begin{minipage}[10pt]{0.95\linewidth} 
  \begin{lstlisting}[language=Pseudocode,mathescape=true,label={lst:our_method}]
  initialize H_proj to empty sparse matrix
  # compute the trust region ratio
  rho = (f(x_prev)-f(x))/(f_tilde(x_prev)-f_tilde(x))
  foreach element i:
     # Pi*x selects element i's local variables from x
     either:
        Hi = constructLocalHessian(i,Pi*x)
        $\Lambda$,U = eig(Hi)
     or:
        $\Lambda$,U = analyticDecomposition(i,Pi*x)
     # our adaptive eigenvalue projection
     $\Lambda$ = (|rho - 1|<epsilon) ? max($\Lambda$,0) : abs($\Lambda$)
     # compute the projected local Hessian
     Hi_proj = U*$\Lambda$*U'
     # accumulate into global Hessian
     H_proj += Pi*Hi_proj*Pi'\end{lstlisting}
\end{minipage}

\begin{figure}[t]
  \centering
  \includegraphics[width=\linewidth]{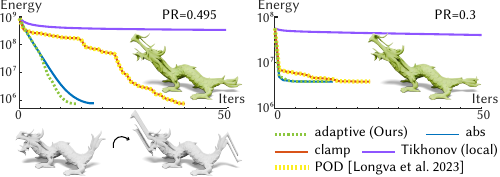}
  \caption{We compare the energy convergence curve in our approach, the Projection-On-Demand (POD)~\cite{longva2023pitfalls} and per-element Tikhonov regularization.
  We also note that, separate from the optimization convergence performance, POD may require additional Cholesky decompositions to check
  the positive definiteness of the Hessian.}
  \label{fig:pod_dragon}
\end{figure}

\section{Results}
\label{sec:result}

\begin{figure}
  \centering
  \includegraphics[width=\linewidth]{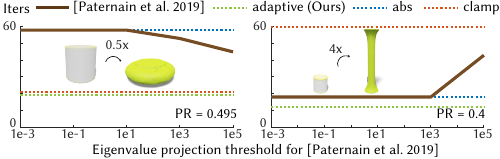}
  \caption[width=0.5\linewidth]{We compare our method to a per-element version of the eigenvalue
  projection scheme in \cite{Paternain2019nonconvex_newton}.  Due to 
  the large variation across local elements' eigenvalue magnitudes
  (under different deformations and Poisson's ratios), 
  it is unclear how to set an eigenvalue projection threshold for \cite{Paternain2019nonconvex_newton}.
  Here, we sweep through a range of threshold values (x-axis) and plot the needed 
  number of iterations in each setting. Our method is consistently 
  more efficient.
  }
  \label{fig:different_hybrid_eps}
\end{figure}

We evaluate our method on a diverse set of examples, including different
deformations, physical parameters, geometries, mesh resolutions, and elastic
energies.
Furthermore, we perform extensive experiments on the TetWild Thingi10k
dataset~\cite{Thingi10K,Hu:2018:TMW:3197517.3201353} to demonstrate our speedup
over existing Hessian projection schemes (\reffig{histogram_thingi10k_large}
and \reffig{histogram_thingi10k_small}).

We implement our method in C++ using the TinyAD library \cite{schmidt2022tinyad} for automatic differentiation, and perform our experiments using a MacBook Pro with an Apple M2 processor and 24GB of RAM.

In our hyperelastic simulation, we use the stable Neo-Hookean energy
\cite{smith2018snh,kim22deformables} with a Young's Modulus of $10^8$ and
different Poisson's ratios.
We run the Projected Newton's method for a maximum of $200$ iterations with a
convergence criteria of $10^{-5}$ on the Newton decrement $0.5 \dx^\top\vg$.
Since the mesh's initial large deformation may invert some finite elements,
we perform a backtracking line search without inversion check to find a suitable
step size\changed{, iteratively shrinking the step size by 0.8 (starting from step size 1) until Armijo condition is satisfied. }
We set the trust-region ratio threshold $\epsilon$ in \refequ{adaptive_eigenvalue_projection} to be $0.01$ for all the examples (except compression for which $\epsilon$ is set to be $0.1$).

\paragraph{Convergence and Stability. }
Our method achieves consistent speedup over existing eigenvalue projection schemes across 
a wide range of deformations (\reffig{stretch_compress} and \reffig{bend_twist}), 
Poisson's ratios (\reffig{different_pr}), mesh resolutions (\reffig{different_resolution}), 
tetrahedralizations (\reffig{mesh_quality}), and 
hyperelastic models (\reffig{different_energy}).
Our approach adaptively \changed{chooses} between
absolute eigenvalue projection and clamping at each iteration.
In the case where a specific eigenvalue projection strategy performs dominantly
better than the other throughout the optimization, our adaptive approach
performs at least comparable to the best-performing eigenvalue projection
strategy.
Note that in the case 
where eigenvalue clamping performs
consistently better than absolute eigenvalue projection, our adaptive
approach can takes one more Newton iteration than eigenvalue clamping
(\reffig{different_resolution} bottom).
This is just because we always start from the absolute eigenvalue projection in the
first Newton iteration. 
Our trust-region based strategy also enables us to obtain more effective descent directions that require much fewer line search iterations to find a suitable step size (\reffig{line_search_iter}).
\changed{We provide the statistics of the trust-region ratio trend, line search iterations, and timings in \reffig{trust_region_ratio_mushroom_stretch}, \reftbl{line_search_iter} and \reftbl{timings}, respectively. }

\begin{table}[t]
  \caption{
    \changed{
    We report the \textbf{average} line search iterations per Newton iteration of the teaser (\reffig{teaser}) and the horse (\reffig{stretch_compress}) example.
    Our adaptive method maintains low line search iterations across different Poisson's ratios and deformation sizes.}
  } 
  \vspace{-2mm}
  \begin{tabular}{ccccccc}
  \whline{1pt}
  \multicolumn{7}{c}{large deformation}                                                                                                                                                      \\ \whline{0.7pt}
  \multicolumn{1}{c|}{}           & \multicolumn{3}{c|}{teaser: \reffig{teaser} (top)}                                     & \multicolumn{3}{c}{horse: \reffig{stretch_compress} (top)}        \\ \cline{2-7}
  \multicolumn{1}{c|}{}           & \multicolumn{1}{c|}{\cellcolor[HTML]{E4FFB1}adaptive} & \multicolumn{1}{c|}{\cellcolor[HTML]{D6FFFD}abs} & \multicolumn{1}{c|}{\cellcolor[HTML]{FFEBE7}clamp} & \multicolumn{1}{c|}{\cellcolor[HTML]{E4FFB1}adaptive} & \multicolumn{1}{c|}{\cellcolor[HTML]{D6FFFD}abs} & {\cellcolor[HTML]{FFEBE7}clamp} \\ \whline{0.5pt}
  \multicolumn{1}{c|}{PR = 0.495} & \multicolumn{1}{c|}{\cellcolor[HTML]{E4FFB1}1.0}      & \multicolumn{1}{c|}{\cellcolor[HTML]{D6FFFD}1.0} & \multicolumn{1}{c|}{\cellcolor[HTML]{FFEBE7}7.5}   & \multicolumn{1}{c|}{\cellcolor[HTML]{E4FFB1}1.8}      & \multicolumn{1}{c|}{\cellcolor[HTML]{D6FFFD}1.0} & {\cellcolor[HTML]{FFEBE7}7.4}   \\
  \multicolumn{1}{c|}{PR = 0.3}   & \multicolumn{1}{c|}{\cellcolor[HTML]{E4FFB1}1.0}      & \multicolumn{1}{c|}{\cellcolor[HTML]{D6FFFD}1.0} & \multicolumn{1}{c|}{\cellcolor[HTML]{FFEBE7}4.0}   & \multicolumn{1}{c|}{\cellcolor[HTML]{E4FFB1}1.5}      & \multicolumn{1}{c|}{\cellcolor[HTML]{D6FFFD}1.0} & {\cellcolor[HTML]{FFEBE7}5.0}  \\ \whline{1pt}
  \multicolumn{7}{c}{small deformation}                                                                                                                                                      \\ \whline{0.7pt}
  \multicolumn{1}{c|}{}           & \multicolumn{3}{c|}{teaser: \reffig{teaser} (bottom)}                                  & \multicolumn{3}{c}{horse: \reffig{stretch_compress} (bottom)}     \\ \cline{2-7}
  \multicolumn{1}{c|}{}           & \multicolumn{1}{c|}{\cellcolor[HTML]{E4FFB1}adaptive} & \multicolumn{1}{c|}{\cellcolor[HTML]{D6FFFD}abs} & \multicolumn{1}{c|}{\cellcolor[HTML]{FFEBE7}clamp} & \multicolumn{1}{c|}{\cellcolor[HTML]{E4FFB1}adaptive} & \multicolumn{1}{c|}{\cellcolor[HTML]{D6FFFD}abs} & {\cellcolor[HTML]{FFEBE7}clamp} \\ \whline{0.5pt}
  \multicolumn{1}{c|}{PR = 0.495} & \multicolumn{1}{c|}{\cellcolor[HTML]{E4FFB1}1.4}      & \multicolumn{1}{c|}{\cellcolor[HTML]{D6FFFD}1.0} & \multicolumn{1}{c|}{\cellcolor[HTML]{FFEBE7}4.1}   & \multicolumn{1}{c|}{\cellcolor[HTML]{E4FFB1}1.0}      & \multicolumn{1}{c|}{\cellcolor[HTML]{D6FFFD}1.0} & {\cellcolor[HTML]{FFEBE7}1.0}   \\
  \multicolumn{1}{c|}{PR = 0.3}   & \multicolumn{1}{c|}{\cellcolor[HTML]{E4FFB1}1.0}      & \multicolumn{1}{c|}{\cellcolor[HTML]{D6FFFD}1.0} & \multicolumn{1}{c|}{\cellcolor[HTML]{FFEBE7}1.0}   & \multicolumn{1}{c|}{\cellcolor[HTML]{E4FFB1}1.0}      & \multicolumn{1}{c|}{\cellcolor[HTML]{D6FFFD}1.0} & {\cellcolor[HTML]{FFEBE7}1.0}   \\ \whline{1pt}
  \end{tabular}
  \label{tbl:line_search_iter}
\end{table}

\begin{table}[t]
  \caption{
    \changed{
    We measure the average per-iteration cost (including all PR and deformation combinations) in wall-clock time on the teaser example ($12k$ vertices, $44k$ tetrahedrons).
    Here we provide the statistics of the average time per Newton iteration, and the percentage of time spent on computing the Newton direction (including the linear solve), line search and (optionally) computing the trust-region (TR) ratio for one Newton iteration.}
  } 
  \vspace{-2mm}
  \begin{tabular}{c|
  >{\columncolor[HTML]{E4FFB1}}c| 
  >{\columncolor[HTML]{D6FFFD}}c| 
  >{\columncolor[HTML]{FFEBE7}}c }
                                                                              & adaptive                                                   & abs                                                        & clamp                                                      \\ \whline{1pt}
  \begin{tabular}[c]{@{}c@{}}Newton \\direction\end{tabular} & \begin{tabular}[c]{@{}c@{}}88.6\% \\ (solve: 76.2\%)\end{tabular} & \begin{tabular}[c]{@{}c@{}}94.6\% \\ (solve: 81.7\%)\end{tabular} & \begin{tabular}[c]{@{}c@{}}80.5\% \\ (solve: 69.4\%)\end{tabular} \\ \hline
  line search                                                                 & 5.4\%                                                      & 5.4\%                                                      & 19.5\%                                                     \\ \hline
  TR ratio                                                          & 6.0\%                                                      & /                                                          & /                                                          \\ \whline{0.75pt}
  \begin{tabular}[c]{@{}c@{}}avg time per \\ Newton iter \end{tabular}   & 0.117 sec                                                    & 0.107 sec                                                     & 0.126 sec                                                      \\ 
\end{tabular}
\label{tbl:timings}
\end{table}

\paragraph{Generality.}
In \reffig{histogram_thingi10k_large} and \reffig{histogram_thingi10k_small}, we
provide the statistics of our speedup over \cite{irving05} and
\cite{chen2024abs_eig} under diverse deformations on 254 high-resolution,
closed, genus-0 tetrahedral meshes with more than 5000 vertices. These meshes are randomly 
chosen from the TetWild Thingi10k dataset~\cite{Thingi10K}.
Our method achieves significant speedup under high Poisson's ratios and large deformations, while still being comparable to the existing eigenvalue projection schemes under small deformations and low Poisson's ratios.

\paragraph{Comparison.} 
Last but not least, we compare our method with another two Hessian projection strategies,
namely, the Projection-On-Demand (POD) strategy~\cite{longva2023pitfalls} in \reffig{pod_dragon} and the per-element Tikhonov
regularization~\cite{Paternain2019nonconvex_newton} in \reffig{different_hybrid_eps}.
As the POD strategy is mostly designed for dynamic simulation where the Hessian
can be close to positive-definite \changed{given a small timestep}, its convergence is
similar to eigenvalue clamping in the quasistatic setting, but it suffers from additional
Cholesky decomposition cost for the positive-definiteness check.
Moreover, Tikhonov regularization~\cite{Paternain2019nonconvex_newton}  requires the eigendecomposition on the global Hessian matrix, 
which is computationally intractable for large meshes. Therefore, we compare our approach to the per-element version of
\cite{Paternain2019nonconvex_newton} (see \reffig{different_hybrid_eps}), where the
local element's eigenvalues are projected to their absolute values when large than a threshold
and clamped otherwise.
As the eigenvalues can have drastically different magnitudes under different
deformations and Poisson's ratios, it is challenging to set a universal
projection threshold for local eigenvalue projection in
\cite{Paternain2019nonconvex_newton}.

\begin{figure}[t]
  \centering
  \includegraphics[width=\linewidth]{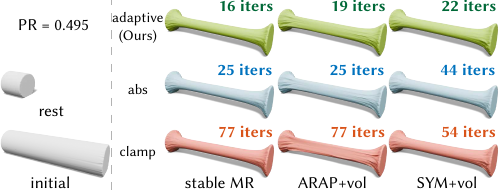}
\caption{Our adaptive strategy generalizes across different Neo-Hookean variants.
Here in the last three columns, we experiment with different strain energies, 
including Mooney-Rivlin, ARAP, and Symmetric Dirichlet energy, each augmented with a volume-preservation term $(J-1)^2$. } 
\label{fig:different_energy}
\end{figure}

\begin{figure}[t]
  \centering
  \includegraphics[width=\linewidth]{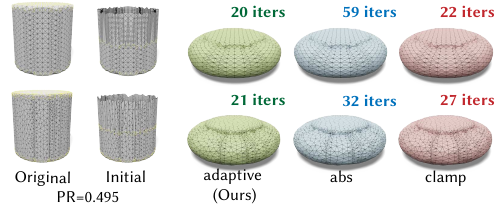}
  \caption[width=0.5\linewidth]{Our method is robust to different tetrahedralizations. 
  Here we tetrahedralize the same cylinder to two different tetrahedral meshes.
  The meshes on the top and bottom row have 6.1k and 6.6k vertices, respectively.
  In both settings, our method outperforms existing strategies.
  } 
  \label{fig:mesh_quality}
\end{figure}

  \vspace{-1mm}

  \vspace{-1mm}

\vspace{-2mm}
\section{CONCLUSION \& FUTURE WORK}
\label{sec:conclusion}

\begin{figure}[t]
  \centering
  \includegraphics[width=\linewidth]{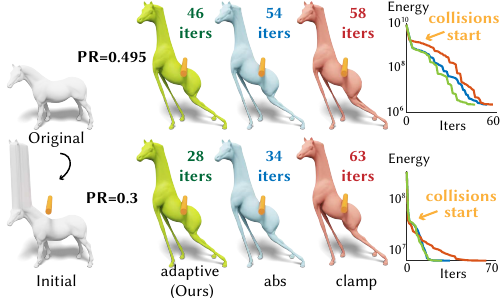}
  \caption{\changed{We perform a collision experiment using Incremental Potential
  Contact (IPC) \cite{Li2020IPC} by placing a cylinder (orange) above the
  back of a horse. 
  Note that IPC's intersection-aware line search dominates convergence and step size after collisions happen.
  Extending our framework to accelerate the convergence of barrier functions could be an interesting future direction. } } 
  \label{fig:horse_collision}
\end{figure}

We introduce a novel adaptive eigenvalue filtering strategy for Projected
Newton's method to stabilize and accelerate the minimization of Neo-Hookean
energy.
Our method is simple to implement and requires only two lines of code change in the Projected Newton framework, making it easy to integrate into existing simulation pipelines.
Our trust-region based framework opens up the possibility of analyzing different eigenvalue projection schemes while taking the quality of the quadratic approximation into account.

We primarily experiment on the quasistatic simulation of stable Neo-Hookean energy in this work.
Extending our framework to other hyperelastic models, collisions (\reffig{horse_collision}) and dynamic simulation, such as combining our work with \cite{longva2023pitfalls}, could be a promising direction.
\changed{
Our choice of $w$ eliminates the need of checking the positive-definiteness of the global Hessian, and is well-suited for quasistatic simulation where the energy is highly nonconvex.
For other more convex scenarios (e.g., dynamic simulation with a small timestep), the global Hessian can sometimes be \emph{serendipitously} positive definite even if some local Hessians are indefinite, in which case setting $w \in [0,0.5]$ could potentially lead to faster convergence \cite{longva2023pitfalls}.}

\changed{Exploration of having $w$ as a continuous function in $[0,1]$ for more fine-grained control} could be another interesting future direction.
\changed{Computing the trust-region ratio and projection strategy independently \emph{for each element} could potentially further improve the convergence.}
Our method requires picking a threshold $\epsilon$ for the trust region ratio.
Further investigation on the choice of this threshold could be beneficial, especially for the case of small compression (see \reffig{histogram_thingi10k_small}).
Our approach always starts from the absolute eigenvalue projection for the first Newton iteration.
Deriving a better strategy for the initial eigenvalue projection could potentially further improve the performance.

\vspace{-2pt}
\begin{acks}
This work is funded in part by two NSERC Discovery grants, the Ontario Early Research Award program, the Canada Research Chairs Program, a Sloan Research Fellowship, the DSI Catalyst Grant program, SoftBank and gifts by Adobe Research and Autodesk.
We thank Silvia Sellán, Otman Benchekroun, Abhishek Madan for help with the rendering of the teaser, all the artists for sharing the 3D models and anonymous reviewers for their helpful comments and suggestions.
\end{acks}

\begin{figure*}[t]
  \includegraphics[width=\linewidth]{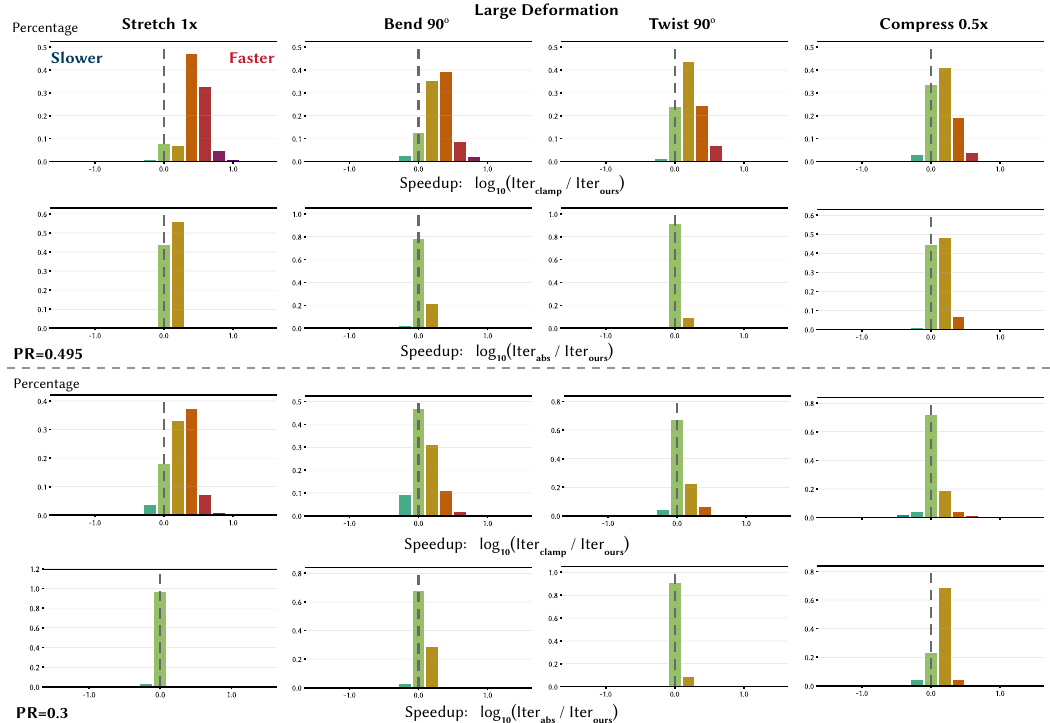}
  \caption{
    Histogram of the speedup of our adaptive method over the absolute 
    eigenvalue projection~\cite{chen2024abs_eig} and eigenvalue clamping~\cite{irving05}
    on TetWild Thingi10k dataset. Here, we test with a diverse set of \emph{large} deformations (shown in each column).
    First row: speedups over eigenvalue clamping with a Poisson's ratio 0.495.
    Second row: speedups over absolute eigenvalue projection with a Poisson's ratio 0.495.
    Third row: speedups over eigenvalue clamping with a Poisson's ratio 0.3.
    Fourth row: speedups over absolute eigenvalue projection with a Poisson's ratio 0.3.
    } 
  \label{fig:histogram_thingi10k_large}
\end{figure*}

\begin{figure*}[t]
  \centering
  \includegraphics[width=\linewidth]{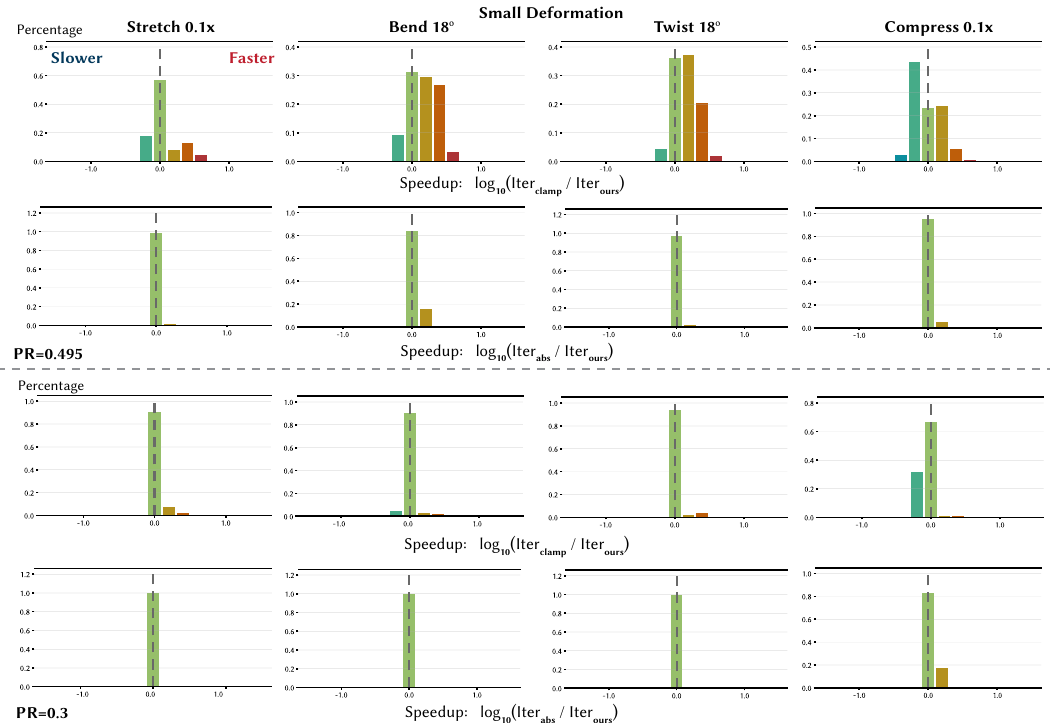}
  \caption{
Histogram of the speedup of our adaptive method in settings
similar to \reffig{histogram_thingi10k_large} but with much \emph{smaller} deformations.
Even in these simple scenarios, our method outperforms eigenvalue clamping under 
a large Poisson's ratio, and performs comparably in other cases.} 
  \label{fig:histogram_thingi10k_small}
\end{figure*}

\bibliographystyle{ACM-Reference-Format}
\bibliography{sections/reference.bib}

\newpage

\appendix
\section{Proof of Lemma 4.1} \label{app:proof_lemma_1}
\begin{lemma}
  Let $|\mA_e|$ be the matrix obtained by performing per-element absolute eigenvalue projection of $\mA$, i.e., $ |\mA_e| = \sum_i \mP_i^{\top} |\mA_i| \mP_i $, where $|\mA_i|$ are the matrix obtained by taking the absolute value of each of the eigenvalues of $\mA_i$.
  Then it holds that $\left|\vx^{\top} \mA \vx\right|  \leq \vx^{\top} |\mA_e| \vx $.
  \label{lem:abs_per_element_constraint}
\end{lemma}
  
\begin{equation}  
  \begin{aligned} 
\left|\vx^{\top} \mA \vx\right| & = \left|\vx^{\top} \left( \sum_i \mP_i^{\top} \mA_i \mP_i \right) \vx\right| \\
& =  \left| \sum_i (\mP_i \vx)^{\top} \mA_i (\mP_i \vx) \right| \\
& = \left| \sum_i \sum_k \left( (\mP_i \vx)^{\top} \ve_{ik} \right) \lambda_{ik} \left( \ve_{ik}^{\top} (\mP_i \vx) \right) \right| \\
& = \left| \sum_i \sum_k \left( (\mP_i \vx)^{\top} \ve_{ik} \right) \lambda_{ik} \left( \ve_{ik}^{\top} (\mP_i \vx) \right) \right| \\
& = \left| \sum_i \sum_k \lambda_{ik} \left( (\mP_i \vx)^{\top} \ve_{ik} \right)^2 \right| \\
& \leq  \sum_i \sum_k \left| \lambda_{ik} \right| \left( (\mP_i \vx)^{\top} \ve_{ik} \right)^2 \\
\nonumber
\end{aligned} 
\end{equation} 

\begin{equation}  
  \begin{aligned} 
& =\sum_i \sum_k \left( (\mP_i \vx)^{\top} \ve_{ik} \right)  \left| \lambda_{ik} \right| \left( \ve_{ik}^{\top} (\mP_i \vx) \right)  \\
& =  \vx^{\top} \left( \sum_i \mP_i^{\top} |\mA_i| \mP_i \right) \vx \\
& = \vx^{\top} |\mA_e| \vx
\nonumber
\end{aligned} 
\end{equation} 
\qed

\section{Proof of Lemma 4.2} \label{app:proof_lemma_3}

\begin{lemma}
  $ \dx^{\top} \mH \dx + \dx^{\top} |\mH_e| \dx = 2 \dx^{\top} \mH_e^+ \dx $.
  \label{lem:hybrid_sum}
\end{lemma}
  
Let $\lambda_{ik}^+$ and $\lambda_{ik}^-$ be the positive and negative eigenvalues of each $\mA_i$.
Then we have:
\begin{equation}  
  \begin{aligned} 
    & \vx^{\top} \mA \vx + \vx^{\top} |\mA_e| \vx \\
    & = \sum_i \sum_k \lambda_{ik} \left( (\mP_i \vx)^{\top} \ve_{ik} \right)^2 + \sum_i \sum_k \left| \lambda_{ik} \right| \left( (\mP_i \vx)^{\top} \ve_{ik} \right)^2 \\
    & = 2 \sum_i \sum_k \lambda_{ik}^+ \left( (\mP_i \vx)^{\top} \ve_{ik} \right)^2 \\
    & = 2 \vx^{\top} \mA_e^+ \vx.
    \nonumber
\end{aligned} 
\end{equation} 
\qed

\end{document}